\documentclass[11pt]{article}

\usepackage{amsfonts}
\usepackage{amsmath}
\usepackage{amssymb}
\usepackage{amsthm}

\usepackage{graphicx}

\theoremstyle{plain}
\newtheorem{theorem}{Theorem}

\newtheorem{corollary}[theorem]{Corollary}

\theoremstyle{definition}

\newtheorem{remark}[theorem]{Remark}

\newcommand{\pdt}{\partial_t}

\begin{document}
\title{{On diffusive stability of Eigen's quasispecies model}}

\author{Alexander S. Bratus$^\textrm{1,2}$, Chin-Kun Hu$^\textrm{3}$, Mikhail V. Safro$^{\textrm{2}}$, Artem S. Novozhilov$^{\textrm{4},}$\footnote{Corresponding author: artem.novozhilov@ndsu.edu}\\[3mm]
\textit{\normalsize $^\textrm{\emph{1}}$Faculty of Computational Mathematics and Cybernetics,}\\[-1mm]
\textit{\normalsize Lomonosov Moscow State University, Moscow 119992, Russia}\\[2mm]
\textit{\normalsize $^\textrm{\emph{2}}$Applied Mathematics--1, Moscow State University of Railway Engineering,}\\[-1mm]\textit{\normalsize Moscow 127994, Russia}\\[2mm]
\textit{\normalsize $^\textrm{\emph{3}}$Institute of Physics, Academia Sinica, Nankang, Taipei 11529, Taiwan}\\[2mm]
\textit{\normalsize $^\textrm{\emph{4}}$Department of Mathematics, North Dakota State University, Fargo, ND 58108, USA}}

\date{}

\maketitle

\begin{abstract}
Eigen's quasispecies system with explicit space and global regulation is considered. Limit behavior and stability of the system in a functional space under perturbations of a diffusion matrix with nonnegative spectrum are investigated. It is proven that if the diffusion matrix has only positive eigenvalues then the solutions of the distributed system converge to the equilibrium solution of the corresponding local dynamical system. These results imply that many of the properties of the quasispecies model, including the critical mutation rates that specify the infamous error threshold, do not change if the spatial interactions under the principle of global regulation are taken into account.

\end{abstract}

\section{Introduction}
Molecular evolution is a complex nonlinear process that requires mathematical modeling to properly estimate mutual influence of the evolutionary forces such as mutation and selection. One of the simplest mathematical models describing the process of molecular evolution is due to Manfred Eigen \cite{eigen1971sma}, see also \cite{baake1999,eigen1989mcc,jainkrug2007}. It is worth noting from the very beginning that the classical models of molecular evolution, including the Eigen model, are usually written as systems of ordinary differential equations (ODE), and therefore describe mean field dynamics. It is well known that explicit spatial structure can significantly change the model behavior, see, e.g., \cite{novozhilov2013replicator,dieckmann2000}; therefore, it is important to evaluate the consequences of spatial variables. We stress that without careful analysis of the resulting systems it is impossible to predict beforehand whether inclusion of spatial variables would imply changes in the system dynamics.

There are a number of studies that offer different approaches to molecular evolution models with spatially dependent interactions (see, e.g., \cite{boerlijst1991sws,boerlijst1993ecs,chacon1995sdm,cronhjort2000ibr,cronhjort1994hvp,czaran2000crp,dieckmann2000,wilke2005quasispecies}). Here we propose a general analytical framework along the lines we considered elsewhere \cite{bratus2011,novozhilov2013replicator,novozhilov2012reaction}, to tackle the Eigen model with explicit spatial structure under global regulation principle. We start with the classical definition of the Eigen model.

Consider a population of individuals (macromolecules) each of which belongs to one of $k$ different types. Assume that the fitnesses (the replication rates) for these types are given by the vector $\alpha=(\alpha_1,\ldots,\alpha_k)$ and denote $A$ the diagonal $k\times k$ matrix with $\alpha_i,\,i=1,\ldots,k$ on the main diagonal. Assume that upon reproduction for a macromolecule of type $i$ it is possible to produce individuals of type $j$, and the probability of this event is $q_{ji}$; hence
$q_{ii} = 1 - \sum_{j \neq i}q_{ji}$ is the probability that no mutations occur at the reproduction. Let $Q$ be the stochastic matrix with entries $q_{ij}$.

If $n_i(t)$ denotes the absolute size of the population of the $i$-th type macromolecules, then, taking into account reproduction events with different fitnesses and mutations, assuming that the time is continuous in the model, after simple bookkeeping we obtain
\begin{equation}\label{eq0:1}
\dot n_i = \sum_{j = 1}^k q_{ij} \alpha_j n_j, \quad i=1,\ldots,k,
\end{equation}
where the dot means the derivative with respect to the variable $t$. System \eqref{eq0:1} can be written in a concise matrix form as
\begin{equation}\label{eq3}
\dot n=QAn=Q_\alpha n,
\end{equation}
where we put $Q_\alpha=QA$ and $n(t)=\bigl(n_1(t),\ldots,n_k(t)\bigr)^\top$. Note that for simplicity we suppress the dependence on $t$ in the differential equations we consider.

It is customary to consider not the absolute sizes of the corresponding types, but the frequencies
$$
w_i(t) = \frac{n_i(t)}{\sum_{j = 1}^kn_j(t)}\,,\quad i=1,\ldots,k,
$$
such that
\begin{equation}\label{eq4}
\sum_{i = 1}^k w_i(t) = 1,
\end{equation}
and vector $w(t)=\bigl(w_1(t),\ldots,w_k(t)\bigr)^\top$ belongs to the simplex $S_k=\{x\in\textbf{R}^k\colon x\geq 0,\,\sum_{i=1}^kx_i=1\}$.
For the new variables we obtain Eigen's quasispecies system \cite{eigen1989mcc}
\begin{equation}\label{eq5}
\dot w_i =\sum_{j=1}^k \alpha_jq_{ij}w_j - w_if^l(t), \quad i=1,\ldots,k,
\end{equation}
with the initial conditions
$$
w(0) = \xi=(\xi_1,\ldots,\xi_k)\in S_k.
$$
Here $f^l(t)$ is the mean population fitness:
\begin{equation}\label{eq6}
f^l(t) = (\alpha, w(t) )=\sum_{i=1}^k \alpha_iw_i.
\end{equation}
Throughout the text we use $(\cdot\,,\cdot)$ to denote the standard inner product in $\textbf R^k$.

The mean fitness satisfies the following nonlinear equation
\begin{equation}\label{eq7}
\frac{df^l}{dt} = (\alpha,\dot w) = (Q_{\alpha}w,\alpha) - \bigl(f^l(t)\bigr)^2,
\end{equation}
where \eqref{eq5} was used.

The only rest point of dynamical system \eqref{eq5} can be found as the solution of the following problem
\begin{equation}\label{eq8}
(Q_{\alpha} - If^l(t))w = 0,
\end{equation}
where $I$ is the identity $k\times k$ matrix.

Due to the fact that the matrix $Q_\alpha$ can be assumed to be primitive, then, invoking the Perron--Frobenius theorem, we have that the unique solution of \eqref{eq8} is given by the positive eigenvector $\hat w = (\hat w_1,\ldots,\hat w_k)^\top\in S_k$ of $Q_\alpha$ that corresponds to the dominant positive eigenvalue $\bar{f}^l = (\hat w,\alpha )$. Note that $\bar{f}^l$ is also the rest point of \eqref{eq7}, because \eqref{eq7} yields that
$(Q_{\alpha}\hat w,\alpha) = \bar{f}^l( \hat{w},\alpha ) =(\bar{f}^l)^2.$

The two major findings of Eigen, who analyzed model \eqref{eq5}, were that the evolution works not on separate types of macromolecules, but on the vectors composed of the mutant clouds; in evolutionary process the winner is not the type with the maximal fitness, as it would happen if all $q_{ij}=0$, but the eigenvector $\hat w$ corresponding to the dominant eigenvalue. This vector was called the \textit{quasispecies}, hence the name of the model. The second important observation is the presence of so-called \textit{error threshold} (for more details see, e.g., \cite{bratus2013linear,wilke2005quasispecies}), which predicts that for many fitness matrices $A$ for the replication rates there exists a critical mutation rate over which the distribution of frequencies of different types of macromolecules becomes uniform. That implies that no statistical method can determine which type has the maximal fitness (for a discussion see, e.g., \cite{baake1999,baake2001mutation,bnp2010,bratus2013linear,bull2005quasispecies,jainkrug2007,saakian2011lethal,saakian2006ese,wilke2005quasispecies}). The discussion on the nature, exact definition and importance of the error threshold for the real populations is ongoing, and we do not plan to touch on these issues here (see our related work in \cite{bratus2013linear}). Instead, we concentrate on the following problem: How to incorporate a spatial structure into the model \eqref{eq5}, and what consequences can be expected if the spatial interactions are modeled.

The rest of the paper is organized as follows. In the next section we discuss our approach to add space to model \eqref{eq5} and state, in a nontechnical language, the main results of our analysis. In Section \ref{sec3} we specify the mathematical problem, the functional spaces in which we are looking for a solution, and discuss in which sense we use the word ``solution.'' Sections \ref{sec4} and \ref{sec5} provides the details of the proof of the main theorem, including some more subtle corollaries. In Section \ref{sec6} we illustrate our findings numerically.

\section{Eigen's quasispecies model with diffusion}
It is a usual approach in mathematical modeling (e.g., \cite{cantrell2003spatial,okubo2002diffusion}) to add the spatial structure to a system of ODE
\begin{equation}\label{eq1:1}
    \dot n=g(n),\quad n\in U\subseteq \textbf R^k,\,g\colon U\rightarrow \textbf R^k
\end{equation}
by considering the following system of partial differential equations (PDE)
\begin{equation}\label{eq1:2}
    \pdt n=g(n)+D\Delta n,\quad n=n(x,t),\,x\in\Omega\subset \textbf R^m,
\end{equation}
where $\pdt n=\frac{\partial n}{\partial t}$, $D=(d_{ij})_{k\times k}$ is the matrix of diffusion coefficients, $m=1,2,$ or $3$ depending on the dimension of the considered physical space (line, plane or space) and $\Delta$ is the Laplace operator. System \eqref{eq1:2} should be supplemented with appropriate initial and boundary conditions. It is well known that the diffusion can be a destabilizing force in a sense that an asymptotically stable rest point of \eqref{eq1:1} can become an unstable spatially uniform rest point of \eqref{eq1:2} through the Turing bifurcation \cite{evans_2010,turing1990chemical}. The exact conditions whether the Turing bifurcation occurs depend on a particular system. Therefore, we are interested in a possibility of a Turing bifurcation for a spatial analogue to system \eqref{eq5}.

An important observation here is that the general approach that leads from \eqref{eq1:1} to \eqref{eq1:2}, however, will not work for the quasispecies model \eqref{eq5}, since $w(t)\in S_k$ represents the vector of frequencies and the dispersal depends on the variation in numbers. The only case when one can add the Laplace operator is to assume that all of the diffusion coefficients are equal, clearly a very restrictive and unrealistic case (see \cite{hadeler1981dfs} for a particular case of such situation).

There are several possible approaches to circumvent the mentioned obstacle and to add explicit spatial structure to \eqref{eq5}, for a short review see \cite{novozhilov2012reaction}. We believe that the correct approach is to follow the lines how the system for the absolute sizes \eqref{eq0:1} was transformed into the system for the frequencies \eqref{eq5}, (this corresponds to the general ideology of the theory of selection systems \cite{karev2010}). Using this approach we arrive (for the details see Section \ref{sec3}) at the equation
\begin{equation}\label{eq1:3}
    \pdt u_i=\sum_{j=1}^k \alpha_jq_{ij}u_j-u_if^s(t)+\sum_j d_{ij}\Delta u_j,\quad x\in\Omega\subseteq\textbf R^m,\quad i=1,\ldots, k
\end{equation}
for the new variable $u(x,t)=\bigl(u_1(x,t),\ldots,u_k(x,t)\bigr)^\top$. In \eqref{eq1:3} we use the notation
\begin{equation}\label{eq1:4}
    f^s(t)=\int_\Omega (\alpha,u)\,dx
\end{equation}
for the mean integral fitness that guarantees that in the domain $\Omega$ the principle of global regulation

\begin{equation}\label{eq1:4a}
\sum_{i=1}^k\int_\Omega u_i(x,t)\,dx=1
\end{equation}
holds for any time moment $t$.

We stress that equation \eqref{eq1:3} is not a usual partial differential equation, but a functional--partial differential equation that includes an integral functional \eqref{eq1:4} and isoperimetric condition \eqref{eq1:4a}. This problem belongs to a not-that-well studied class of mathematical problems.

This approach was used in \cite{bratus2006ssc,bratus2009existence,bratus2011} for various classes of replicator equations. For the Eigen model \eqref{eq5} this approach in the case when space has one dimension was applied for the first time in \cite{weinberger1991ssa}, and it was argued that the spatially uniform steady state for the local model \eqref{eq5} remains asymptotically stable for the case of system \eqref{eq1:3} when the diffusion matrix is diagonal. We generalize this statement for an arbitrary dimension of the physical space and any diffusion matrix $D$ that has positive eigenvalues. The main statement is given below, and the rest of the paper is devoted to the proof of it.

\begin{theorem}\label{th:1} Consider simultaneously the local Eigen quasispecies model \eqref{eq5} and the spatially explicit model \eqref{eq1:3} with zero Neumann boundary conditions (the boundary of $\Omega$ is impenetrable). Let domain $\Omega$ be bounded and have a piecewise smooth boundary, and diffusion matrix $D$ have only distinct positive eigenvalues. Then
\begin{enumerate}
\item There exits only one steady state solution to \eqref{eq1:3}, and it coincides with the rest point $\hat w$ of~\eqref{eq5};
\item This steady state solution $\hat w$ is globally stable.
\end{enumerate}
\end{theorem}

\begin{remark}
The convergence for the solutions of system \eqref{eq1:3} to the rest point $\hat{w}$ is understood as the convergence in an appropriate functional space, the details are presented in Sections \ref{sec3} and \ref{sec4}.
\end{remark}

\begin{remark}
Model \eqref{eq5} is not the only possible mathematical reincarnation of the Eigen quasispecies model. Another, and somewhat more justified approach to deduce a model that is described by a system of ODE is to assume that the mutations and reproduction events are separated, and reproduction proceeds error-free. In this case, it is possible to show \cite{baake1999} that the system takes the form
\begin{equation}\label{eq1:5}
    \dot w_i=(m_i-f^l(t))w_i+\sum_{i=1}^k\mu_{ij}w_j,\quad i=1,\ldots,k,
\end{equation}
where now $m_i$ are Malthusian fitnesses (the reproduction rates) and $\mu_{ij}$ are the mutation rates, $\mu_{ii}=-\sum_{j\neq i}\mu_{ji}$, and the mean fitness is given by
$$
f^l(t)=\sum_{i=1}^km_iw_i.
$$
Systems \eqref{eq5} and \eqref{eq1:5} have very close properties \cite{hofbauer1985selection}. Theorem \ref{th:1} (with necessary modifications) remains valid for system \eqref{eq1:5} as well.
\end{remark}

\section{Notation and derivation of the Eigen model with diffusion}\label{sec3}
In this section we present a general approach to derive a spatially explicit counterpart of the classical Eigen's quasispecies model \eqref{eq5}.

Let $\Omega \subset \textbf{R}^m,\, m=1,2,3$ be a bounded domain with a piecewise-smooth boundary $\gamma$. Denote $n(x,t)$ vector-function that describes absolute sizes of different macromolecules composing the system; in coordinates $n(x,t) = \bigl(n_1(x,t),\ldots,n_k(x,t)\bigr)^\top$, $t \geq 0$, $x \in \Omega$ in space $\textbf{R}^m$. Consider matrix $D=(d_{ij})_{k\times k}$ with positive eigenvalues $\mu_i>0$, $i=1,\ldots,k$ and assume that there exists a matrix $U$ such that
\begin{equation}\label{eq9}
D=UM U^{-1},
\end{equation}
where $M$ is a diagonal matrix with $\mu_i$ on the main diagonal. The diffusion matrix $D$ describes the action of spatial diffusion on macromolecules; note that we do not require matrix $D$ be diagonal, and hence include in consideration possible cross diffusion interactions that describe a directed movement of macromolecules with respect to an external signal (which can be, e.g., chemical, see for more detail \cite{berezovskaya2008families,hillen2009user}).

Following the standard approach (e.g., \cite{cantrell2003spatial}) the dynamics of the populations of macromolecules can be described with help of the following PDE
\begin{equation}\label{eq10}
\pdt n_i(x,t) = (Q_\alpha n(x,t))_i + (D \Delta n(x,t))_i,\quad i=1,\ldots,k.
\end{equation}
Here $\Delta$ is the Laplace operator, in Cartesian coordinates
$$
\Delta = \sum_{i=1}^m \frac{\partial^2}{\partial x_i^2}\,,\quad m=1,2,3,
$$
$Q_\alpha$ is the same matrix as in equality \eqref{eq3}, and the following notations are used in the rest of the paper:
\begin{equation*}
\begin{split}
(Q_\alpha n(x,t))_i &= \sum\nolimits_{j=1}^kq_{ij}\alpha_{j}n_{j}(x,t),\\
(D \Delta n(x,t))_i &= \sum\nolimits_{j = 1}^{k}d_{ij}\Delta n_j(x,t).
\end{split}
\end{equation*}
The initial data for the concentrations of macromolecules are given by
$$
n_i(x , 0) = \Psi_i(x), \quad i=1,\ldots,k.
$$
It is natural to assume that we consider a closed system, i.e., we have zero flow boundary condition (the Neumann boundary conditions)
$$
\left.\frac{\partial n_i(x,t)}{\partial \nu} \right|_{x\in\gamma} = 0, \quad i=1,\ldots,k,
$$
where $\nu$ is the outward normal to the boundary $\gamma$.

We would like to rewrite system \eqref{eq10} for the relative concentrations. Consider
\begin{equation}\label{eq11}
u_i(x,t) = \frac{n_i(x,t)}{\sum_{j=1}^k \int_{\Omega} n_j(x,t)\, dx}\,, \quad i=1,\ldots,k.
\end{equation}
We obtain equality \eqref{eq4} in the form of the following integral invariant
\begin{equation}
\label{eq12}
\langle(u,1_k)\rangle = \sum \limits_{i=1}^k \int \limits_{\Omega} u_i(x,t)dx=1,
\end{equation}
where the following notation
$$
\langle (u,1_k) \rangle =\int_{\Omega} (n(x,t),1_k)\,dx, \quad 1_k = (1,\ldots,1)\in \textbf R^k
$$
was used; in words, we use the angle brackets $\langle F \rangle$ to denote the integral of $F$ over $\Omega$, and the parentheses, or round brackets $(\cdot\,,\cdot)$ to denote the standard inner product in $\textbf{R}^k$.

From (\ref{eq11}) it follows that
\begin{equation*}
\begin{split}
\pdt u_i &= \frac{\pdt n_i \langle(n,1_k)\rangle - n_i \langle (\pdt n,1_k  )\rangle}{\langle(n,1_k)\rangle^2} =\\
&=\frac{\left((Q_{\alpha}n)_i + (D\Delta n)_i \right)\langle(n,1_k)\rangle - n_i \langle(Q_{\alpha}n + D\Delta n, 1_k)\rangle}{\langle(n,1_k)\rangle^2}\,,\quad i=1,\ldots,k.
\end{split}
\end{equation*}
After simplification, we obtain the following equations
\begin{equation}\label{eq13}
\pdt u_i = (Q_\alpha u)_i -u_i f^s(t) + (D \Delta u)_i,\quad i=1,\ldots,k,
\end{equation}
with the initial and boundary conditions
\begin{equation}\label{eq14}
u_i(x,0) = \phi_i(x), \quad \left.\frac{\partial u_i(x,t)}{\partial \nu} \right|_{x\in\gamma} = 0,\quad t\geq 0,\, i = 1,\ldots,k .
\end{equation}
Function $f^s(t)$ is the average integral fitness of the system and is given explicitly by
\begin{equation}\label{eq15}
f^s(t) = \langle (u(x,t), \alpha) \rangle.
\end{equation}
For obtaining the last expression we used boundary conditions (\ref{eq14}) and Green's formula
$$
\int_{\Omega} \Delta u_i(x,t)dx = \int_{\gamma} \frac{\partial u_i(x,t)}{\partial \nu} ds = 0, \quad i=1,\ldots,k.
$$

Suppose that for any fixed $t$ each function $u_i(x,t)$ is differentiable with respect to $t$ and belongs to the Sobolev space $W^{2,m}(\Omega)$ if $m=1,2$ or $W^{2,2}(\Omega)$ if $m=3$ as the function of $x$ for any fixed $t \geq 0$ . Here $W^{m,2}$, $m=1,2$ is the space of functions, which have square integrable derivatives with respect to variables $x_i$ up to the order $m$. We require this condition because due to the embedding theorem we have that any function from the space $W^{s,2}(\Omega)$ is continuous, except possibly on the set of measure zero (see, e.g., \cite{evans_2010} for details).

Denote $\Omega_t = \Omega \times [0,+\infty)$ and consider the space $B(\Omega_t)$ of functions with the norm
$$
\|v\|_{B} = \max_{t>0}\left\{\|v\|_{W^{s,2}} + \|\pdt v\|_{W^{s,2}}\right\}.
$$
Denote $S_k\left(\Omega_t\right)$ the set of non-negative vector-functions $u(x,t),$ such that $u_i(x,t)\in B\left(\Omega_t\right),\,i=1,\ldots,k$ that satisfy (\ref{eq12}).
In the following we look for a weak solution to (\ref{eq13}), (\ref{eq14}), i.e., the solutions that satisfy the integral identity
$$
\int_{0}^{\infty}\int_{\Omega}\pdt u_i\eta(x,t)dxdt = \int_{0}^{\infty}\int_{\Omega}(Q_{\alpha}u(x,t))_i \eta(x,t) dxdt - \int_{0}^{\infty}\int_{\Omega}(D \nabla u, \nabla \eta)_idxdt
$$
for any function $\eta(x,t)$ on compact support in variable $t$, which is differentiable on $[0, +\infty)$ with respect to $t$ and belongs to the Sobolev space $W^{s,2}(\Omega)$ for any fixed $t \geq 0,\, s=1,2$.

We remark that system (\ref{eq13}) is not a usual system of PDE because its right hand side contains the functional $f^s(t)$ on solutions $u_i(x,t)$. Nevertheless, it is possible to consider this system in the classical sense if we add the following ODE for the average integral fitness that can be obtained similarly to equation (\ref{eq7})
\begin{equation}
\label{eq16}
\frac{df^s}{dt}(t) = \langle(Q_{\alpha}u, \alpha) \rangle - \bigl(f^s(t)\bigr)^2 .
\end{equation}
Steady state solutions to (\ref{eq13}) can be found as solutions to the following PDE of elliptic type:
\begin{equation}\label{eq17}
\left(D \Delta v(x) \right)_i + (Q_{\alpha}v(x))_i - v_i(x) \bar{f}^s = 0,\quad i=1,\ldots,k,
\end{equation}
with the boundary conditions
\begin{equation}
\label{eq18}
\left.\frac{\partial v_i(x)}{\partial \nu} \right|_{x\in\gamma} = 0, \quad i = 1,\ldots,k.
\end{equation}
The integral invariant \eqref{eq12} now reads
\begin{equation}\label{eq19}
\langle (v(x),1_k) \rangle = \sum_{i=1}^k \int_{\Omega} v_i(x)dx = 1.
\end{equation}
Denote by $S_k(\Omega)$ the set of all non-negative vector functions $v(x) = \bigl(v_1(x),\ldots,v_k(x)\bigr)^\top$, $v_i(x) \in W^{s,2}(\Omega)$ ($s=1$ or $s=2$) that satisfy (\ref{eq19}). Using (\ref{eq15}) and (\ref{eq19}) we obtain that
$$
\bar{f}^s = \langle(v(x), \alpha) \rangle,
$$
i.e., $\bar{f}^s$ is a constant.

\section{Steady state solution to the distributed Eigen's system}\label{sec4}
In this section we prove the first part of Theorem \ref{th:1}.

\begin{theorem}\label{theorem1}
If the diffusion matrix $D$ has only positive eigenvalues $\mu_i$, $i=1,...,n$, then the steady state solution to the system \eqref{eq17}, \eqref{eq18} in $S_k\left(\Omega \right)$ coincides with the rest point $\hat{w}$ of the local system \eqref{eq5}.
\end{theorem}
\begin{proof} Consider the following boundary eigenvalue problem
\begin{equation}\label{eq20}
\Delta \psi(x) = -\lambda \psi(x), \quad \left.\frac{\partial \psi}{\partial \nu} \right|_{x\in \gamma} = 0.
\end{equation}
The eigenfunction system is given by $\psi_0(x) = 1, \{\psi_j(x) \}_{j=1}^{\infty}$ and forms a complete system in the Sobolev space $W^{s,2} \left(\Omega \right)$; moreover, we have
\begin{equation}\label{eq21}
\int_{\Omega}\psi_i(x)\psi_j(x)\,dx = \delta_{ij},
\end{equation}
where $\delta_{ij}$ is the Kronecker symbol. The corresponding eigenvalues are known to satisfy
$$\lambda_0 = 0 < \lambda_1 \leq \lambda_2 \leq ...\leq\lambda_j \leq \ldots,\quad  \lim_{j \to \infty}\lambda_j =+\infty.$$
Consider solution $v(x) = \bigl(v_1(x),\ldots,v_k(x)\bigr)^\top\in S_k(\Omega)$ to the system \eqref{eq17}, \eqref{eq18}. Assume that this solution has the following form
\begin{equation}\label{eq22}
v_i(x) = c_i^0 + \sum_{j=1}^{\infty}c_i^j \psi_j(x), \quad i = 1,\ldots,k,
\end{equation}
which is always possible since the eigenfunctions $\{\psi_j(x) \}_{j=0}^{\infty}$ of \eqref{eq20} form a complete system in $W^{s,2}(\Omega)$. Here $c_i^0,c_i^j$ are constants, $i=1,\ldots, k$, $j = 1,2,\ldots$

We substitute \eqref{eq22} into \eqref{eq17} and multiply consecutively the resulting expression by the eigenfunctions of \eqref{eq20}. After integrating over $\Omega$ we obtain a system of linear equations for $c_i^j$
\begin{equation}\label{eq23}
(Q_{\alpha}C^0)_i - \bar{f}^sc_i^0 =0, \quad i = 1,\ldots,k,\\
\end{equation}
\begin{equation}\label{eq24}
(Q_{\alpha}C^j)_i - \bar{f}^s c_i^j - (D \lambda_j C^j)_i=0, \quad j = 1,2,\ldots
\end{equation}
where $\lambda_j$ are the eigenvalues of \eqref{eq20}, $C^{0} =(c_1^0,\ldots,c_k^0)$ and $C^j = (c_1^j,\ldots,c_k^j)$ are vectors in $\textbf{R}^k$, $j=1,2,\ldots$ Here we essentially used the orthogonality of the eigenfunctions $\psi_j(x)$.

The integral invariant \eqref{eq19} and the average integral fitness take the form
$$
\sum_{i=1}^k c_i^0 = 1, \quad \bar{f}^s = \sum_{i=1}^k \alpha_i c_i^0.
$$
Equation \eqref{eq23} has the same form as the equation for the rest point of system \eqref{eq5}. Therefore, the unique solution of that equation coincides with the coordinates of rest point, i.e. $C^0 = \hat{w}$, which implies
$$
\bar{f}^s = \bar{f}^l.
$$

Now consider system \eqref{eq24}. Since $D$ satisfies \eqref{eq9}, then there exists a non-singular matrix $U$ such that $U^{-1}DU = M$, where $M$ is a diagonal matrix with the entries $\mu_1, \ldots,\mu_k$ on the main diagonal.
Consider the transformation
$$
C^j = U Y^j, \quad j=1,2,\ldots
$$
where $Y^j=(y_1^j,...,y_k^j)$ is a vector in $\textbf{R}^k$. Equations \eqref{eq24} take the following form
\begin{equation}\label{eq25}
R_{\alpha}Y^j = (\bar{f}^l I + \lambda_j M )Y^j, \quad j=1,2,\ldots
\end{equation}
The eigenvalues of the matrix $R_{\alpha}=U^{-1}Q_{\alpha}U$ and matrix $Q_{\alpha}$ coincide.

Using the fact that $\bar{f}^l$ is the maximal eigenvalue of the matrix $Q_{\alpha}$ we obtain
$$
((U^{-1}Q_{\alpha}U )Y^j,Y^j ) \leq \bar{f}^l\|Y^j\|^2.
$$
By the assumptions, the eigenvalues $\mu_i$ of the diffusion matrix $D$ are positive and hence
$$
((\bar{f}^l I + \lambda_jM)Y^j, Y^j) > \bar{f}^l \|Y^j\|^2.
$$
Therefore, equation \eqref{eq25} has no nontrivial solutions. It implies that the only solution to equation \eqref{eq23} is exactly the rest point of the local system \eqref{eq5}.
\end{proof}

\begin{remark}
The case when at least one of the eigenvalues of $D$ is zero is more intricate. Assume, without loss of generality, that $\mu_1=0$, then it is possible that
$$
(MY^j,Y^j )=0,\quad j=1,2,\ldots.
$$
Therefore, the equation \eqref{eq25} can have a countable set of solutions.

Consider the following example. Suppose that the rest point of \eqref{eq5} is an eigenvector of $D$ with the zero eigenvalue: $D\hat{w}=0$. Then the vector functions
$$
v^j(x) = \frac{1}{\left|\Omega\right|}\hat{w}\left(1 + \psi_j(x)\right),\quad j=1,2,\ldots
$$
are the solutions to the system \eqref{eq17},\eqref{eq18}. Here $\left|\Omega\right|$ is the measure of $\Omega$. Indeed, we have
\begin{equation*}
\begin{split}
D\Delta v^j + Q_{\alpha}v^j - I\bar{f}^lv^j &= \frac{1}{|\Omega|}(D \hat{w} )\psi^j(x) = 0,\\
\frac{1}{|\Omega|}\sum_{i=1}^k \int_{\Omega}v_i^j(x)dx &= \frac{1}{|\Omega|}\sum_{i=1}^k \int_{\Omega}\hat{w}_idx=1.
\end{split}
\end{equation*}
\end{remark}

\section{Stability of the spatially homogeneous solution to the distributed Eigen's system}\label{sec5}
In this section we provide a proof to the second part of Theorem \ref{th:1}.

For the following we will need the statement: If the system
$\dot x = Ax(t)$ with a constant matrix $A$ has bounded solutions when $t\to\infty$ (i.e., the origin is stable) then the system
$\dot y = ( A + B(t)) y(t)$ also has bounded solutions if $\int_{0}^\infty\|B(t)\|dt<\infty$; the proof can be found in, e.g.,~\cite{verhulst1996nonlinear}.

\begin{theorem}
\label{theorem2}
If the eigenvalues of the diffusion matrix $D$ are positive then for any initial data $\phi_i(x) \in S_n\left( \Omega\right)$ solutions to the system \eqref{eq13}, \eqref{eq14} converge in $S_n(\Omega_t)$ to the solutions of the system \eqref{eq5} with the initial data
\begin{equation}
\label{eq26}
\xi_i = \int_{\Omega}\phi_i(x)dx,\quad i=1,\ldots,k.
\end{equation}
\end{theorem}
\begin{proof}As in the proof of Theorem \ref{theorem1} we look for the solution $u_i(x,t) \in S_n\left(\Omega_t\right)$ of \eqref{eq13},\eqref{eq14} in the form of the following expansion
\begin{equation}
\label{eq27}
u_i(x,t) = c_i^0(t) + \sum_{j=1}^{\infty}c_i^j(t)\psi_j(x), \quad i = 1,\ldots,k.
\end{equation}
Here $c_i^j(t)$ are functions of $t$, $i = 1,\ldots,k$, $ j = 1,2,\ldots$. This expansion is possible because the system of eigenfunctions $\{\psi_j(x) \}_{j=0}^{\infty}$ forms a complete system in $W^{s,2}(\Omega),\, s=1,2$. We obtain the following system of ODE for the unknowns $c_i^j(t)$
\begin{equation}\label{eq28}
\frac{d c_i^0}{dt}(t) = (Q_{\alpha}c^0(t))_i - f^s(t)c_i^0(t),\\
\end{equation}
\begin{equation}
\label{eq29}
\frac{dc_{i}^{j}}{dt}(t) = (Q_{\alpha}c^j(t))_i - f^s(t) c_i^j(t) - (D \lambda_jc^j(t))_i,
\end{equation}
where $\lambda_j$ are the eigenvalues of \eqref{eq20}; $c^0(t) = (c_1^0(t),\ldots, c_k^0(t) )$, $c^j(t) = (c_1^j(t),\ldots, c_n^j(t) ).$
The integral fitness \eqref{eq15} has the form
$$
f^s(t) = \sum_{i=1}^k \alpha_i c_i^0(t).
$$
Since
$$
\int_{\Omega}\phi_i(x)dx = \int_{\Omega}u_i(x,0)dx = \xi_i,
$$
then $c_i^0 = \xi_i$, $i = 1,\ldots,k$. From the uniqueness theorem it follows that the solution to the systems \eqref{eq5} and \eqref{eq28} coincide, i.e. $c_i^0(t) = w_i(t)$ and therefore
\begin{equation}\label{eq30}
f^s(t) = f^l(t).
\end{equation}
Now consider system \eqref{eq29}. As in the proof of Theorem \ref{theorem1} consider the transformation $U$ which reduces the diffusion matrix to the diagonal form. System \eqref{eq29} takes the form
$$
\frac{dY^j}{dt} = (U^{-1}Q_{\alpha}U)Y^j(t) - (f^l(t)I + \lambda_jM)Y^j(t),\quad j=1,2,\ldots
$$
where $M$ is the diagonal matrix. Using
$$
f^l(t) = (f^l(t) - \bar{f}^l) + \bar{f}^l,
$$
we obtain
$$
\frac{dY^j}{dt} =AY^j(t) + B(t)Y^j(t),
$$
where $A$ is the constant matrix
$$
A = (U^{-1}Q_{\alpha}U) - (\bar{f}^sI + \lambda_jM),
$$
$B(t)$ is a diagonal matrix with the elements
$$
b_{ii}(t) = \bar{f}^l - f^l(t),\quad i=1,\ldots,k.
$$
Recall that $f^l(t)$ is given by \eqref{eq6} and $\bar{f}^l=(\hat{w},\alpha)$.
To investigate the system's behavior for
$$
\frac{dX^j}{dt} =AX^j(t), \quad X^j(t) \in \textbf{R}^k
$$
we introduce the following Lyapunov functions
$$
V(X^j ) = \frac{1}{2}(X^j,X^j)= \frac 12\|X^j\|^2.
$$
One has
$$
\dot{V}(X^j) = (AX^j,X^j) =( (U^{-1}Q_{\alpha}U )X^j,X^j) - (( f^lI + \lambda_jM)X^j,X^j).
$$
On the other hand,
\begin{equation*}
\begin{split}
((U^{-1} Q_{\alpha}U )X^j, X^j) &\leq \bar{f}^l \| X^j \|^2,\\
((f^lI + \lambda_jM)X^j, X^j) &> \bar{f}^l\| X^j\|^2.\\
\end{split}
\end{equation*}
Therefore the equilibrium $X^j = 0$ is asymptotically and exponentially stable.

Now estimate the following integral
\begin{equation}\label{eq31}
\int_{0}^{\infty}|\bar{f}^l - f^l(t)|\,dt.
\end{equation}
It can be reduced to
$$
\int_{0}^{\infty}|\hat{w}_i - w_i(t)|\,dt, \quad i = 1,\ldots,k,
$$
where $w_i(t)$ are the coordinates of the solution to the equation \eqref{eq5} and $\hat{w}_i$ are the coordinates of the rest point of \eqref{eq5}.

Using the well known results for the convergence of the solutions to the rest point of \eqref{eq5} (e.g., \cite{jones1976tsc})  we obtain
that there exist positive constants $K$ and $\kappa$ such that
$$|\hat{w}_i - w_i(t)| \leq K e^{-\kappa t}.$$
Therefore the integral \eqref{eq31} converges. It implies that the solutions to the system \eqref{eq29} tend to zero when $t \to \infty$. This completes the proof.
\end{proof}

\begin{corollary}
If at least one eigenvalue of the diffusion matrix $D$ is equal to zero then it is possible to have a countable set of solution $u^j(x,t) \in S_n\left(\Omega_t\right)$, $j=1,2,\ldots$ to the system \eqref{eq13}, \eqref{eq14}. These solutions converge to the solutions of the dynamical system \eqref{eq5} in the sense of average integral value, i.e.,
$$
\lim_{t \to \infty}\int_{\Omega}u_i^j(x,t)dx \to {w}_i(t), \quad i=1,...,k, \quad j=1,2,\ldots
$$
where $w(t)$ is the solution of the system \eqref{eq5}.
\end{corollary}

Indeed in this case we cannot guarantee convergence the solutions of \eqref{eq29} to zero. Consider a countable set of vector functions
$$
u^j(x,t) = \frac{1}{|\Omega|}(\omega(t)+\hat{\omega}\psi_j(x)),
$$
\begin{equation*}
\begin{split}
DU^j(x,t)+Q_{\alpha}U^j(x,t) - I\bar{f}^lu^j(x,t) - \pdt u^j &=\\
\frac{1}{|\Omega|}((Q_{\alpha}\omega(t) - I\bar{f}^l\omega(t) - \pdt w)+(D \hat\omega + Q_{\alpha}\hat \omega-I\bar{f}^l\hat \omega)\psi_j(x)) &=0,
\end{split}
\end{equation*}
since $D\hat{\omega}=0$. Additionally,
$$
\sum_{i=1}^{k} \int_{\Omega}u_i^j(x,t)dx=\frac{1}{|\Omega|}\sum_{i=1}^{k} \int_{\Omega} \omega_i(t)dx = \sum_{i=1}^{k}\omega_i(t) = 1.
$$

\section{Numerical calculations}\label{sec6}

We illustrate the results of analytical calculations with the following example. Consider the population of macromolecules that have length $L=4$ and composed of ones and zeros. Assuming that the corresponding fitnesses depend only on the amount of zeroes and not on their positions (the fitness landscape is \textit{permutation invariant}), the mutations are independent of the site position and occur with probability $q$ per individual per position per replication, then the whole quasispecies system can be rewritten as $L+1$ dimensional system with the variables $u(x,t)=(u_1,\ldots,u_5)$.

Numerical solutions for such system are shown in Fig.~\ref{fig1}. In particular, the top panel gives the time evolution of the spatially independent frequencies $\bar{u}_i(t)=\int_\Omega u_i(x,t)dt$, and the bottom panels show how spatially non-homogeneous initial conditions become homogeneous under the influence of spatial diffusion. The overall picture is in agreement with the analytical investigation from the previous sections; the solutions approach the values $\hat{w}$ of the rest point of the local Eigen system, which can be found as the coordinate of the positive eigenvector.

\begin{figure}[!t]
\begin{center}
\includegraphics[width=0.75\textwidth]{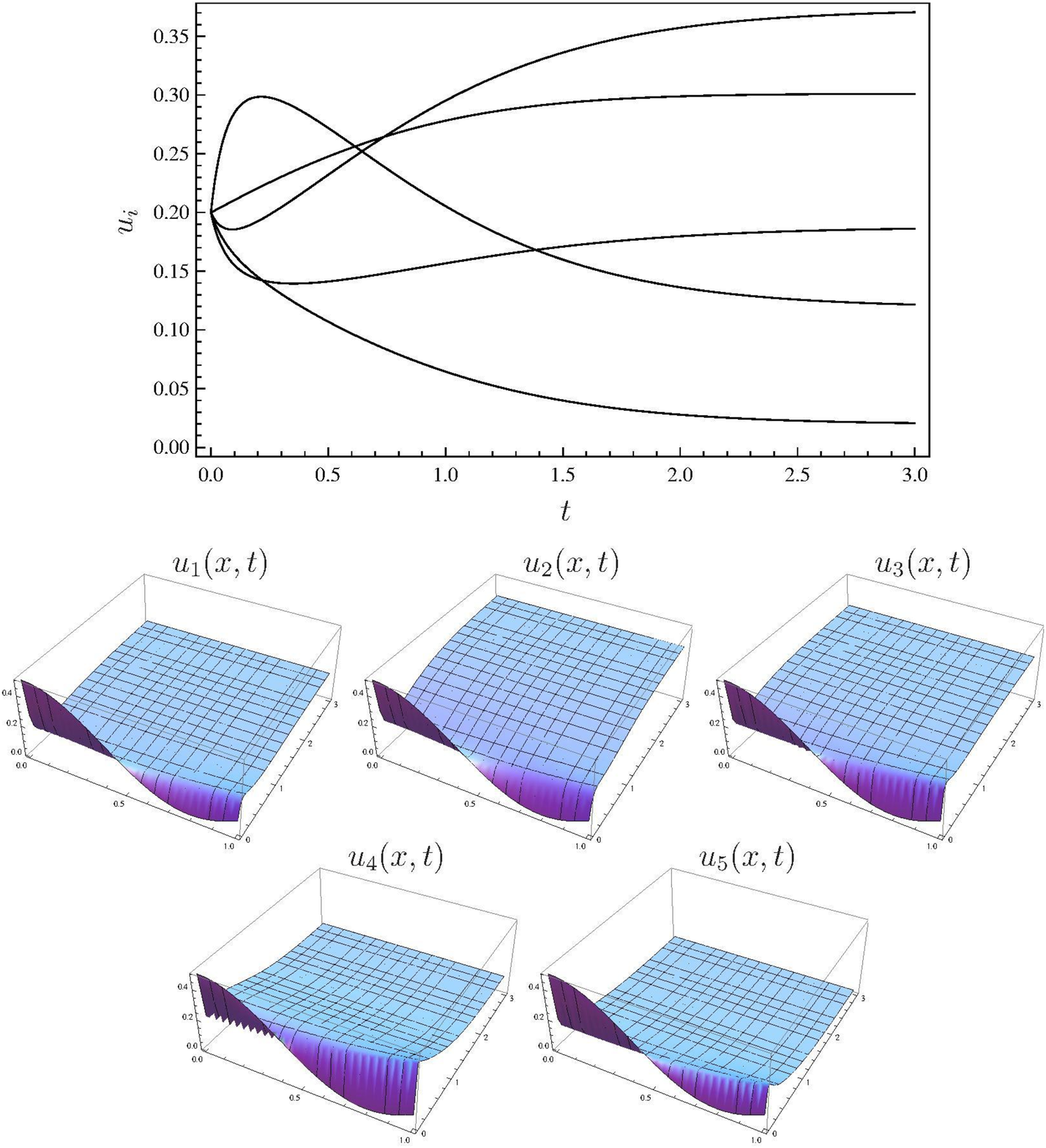}\\
\end{center}
\caption{Numerical solutions to the distributed Eigen's quasispecies model (see text for details)}\label{fig1}
\end{figure}

\paragraph{Acknowledgements:} This research is supported in part by the Russian Foundation for Basic Research (RFBR) grant \#10-01-00374 and joint
grant between RFBR and Taiwan National Council \#12-01-92004HHC-a. ASN's research is supported in part by ND EPSCoR and NSF grant \#EPS-0814442.


\end{document}